\documentclass[letterpaper,superscriptaddress,showpacs,floatfix]{iopart} 
\usepackage{graphicx,helvet}
\usepackage{color}
\usepackage{bm}                      
\usepackage{tikz}
\usepackage{amsthm}
 \usepackage{soul}
\usepackage{amsfonts}
\definecolor{mygray}{gray}{0.4}
\definecolor{light-blue}{rgb}{0.8,0.85,1}
\graphicspath{{images/}}
\usepackage{wasysym}%
\usetikzlibrary{decorations.shapes}
\usepackage[draft,inline,nomargin]{fixme} \fxsetup{theme=color} 

\usepackage[]{lineno} 
\newcommand{\mcU}{\mathcal{U}}
\newcommand{\mcO}{\mathcal{O}}

\newcommand{\mcL}{\mathcal{L}}

\newcommand{\mcN}{\mathcal{N}}
\newtheorem{theorem}{Theorem}
\newtheorem{definition}{Definition}

\newtheorem{corollary}{Corollary}
\newtheorem{example}{Example}
\newtheorem{lemma}{Lemma}

\newcommand{\mcG}{\mathcal{G}}

\newcommand{\ie}{i.e.}

\newcommand{\gqc}{GQC}
\newcommand{\sgqc}{SGQC}
\newcommand{\dgqc}{$\delta\text{GQC}$}
\newcommand{\gf}{GF}
\newcommand{\tp}{TP}
\newcommand{\hp}{HP}
\newcommand{\Imi}{\imath}

\renewcommand{\eref}[1]{\Eref{#1}}

\newcommand{\gs}{GS}
\newcommand{\cptp}{CPTP}

\newcommand{\text}[1]{\mathrm{#1}}

\newcommand{\unam}{Universidad Nacional Aut\'onoma de M\'exico, Ciudad de M\'exico 01000, M\'exico}
\newcommand{\ifunam}{Instituto de F\'{\i}sica, \unam}

\newcommand{\regens}{Institut f\"ur Theoretische Physik, Universit\"at Regensburg, D-93040 Regensburg, Germany}

\begin{document}
\title[]{Position representation of single-mode Gaussian channels beyond the Gaussian functional form}

\author{David Davalos$^1$, Camilo Moreno$^2$, Juan-Diego Urbina$^2$, Carlos Pineda$^1$} 
\address{$^1$ \ifunam}
\address{$^2$ \regens}
\ead{davidphysdavalos@gmail.com}

\begin{abstract} 
We study one-mode Gaussian quantum channels in continuous-variable systems by performing a
\textit{black-box} characterization using complete positivity and trace
preserving conditions, and report the existence of two subsets that do not have a functional 
Gaussian form. Our study covers as particular limit the case of singular channels, thus connecting our results
with their known classification scheme based on canonical forms. Our full characterization of Gaussian channels without Gaussian functional form is completed by showing how
Gaussian states are transformed under these operations, and by deriving the conditions for the existence 
of master equations for the non-singular cases. We show that although every functional form can be found in the vicinity of the identity, one of them does not parametrize unitary channels.
\end{abstract} 
 
\maketitle
\section{Introduction} 
Within the theory of continuous-variable quantum systems (a central topic of study given their role in the description of physical systems like the electromagnetic
field~\cite{doi:10.1142/p489}, solids and nano-mechanical
systems~\cite{RevModPhys.86.1391} and atomic
ensembles~\cite{RevModPhys.82.1041}) the simplest states, both from a theoretical an experimental point of view, 
are the so-called Gaussian states. An operation that transforms such
family of states into itself is called a Gaussian quantum channel (\gqc{}).
Even though Gaussian states and
channels form small subsets among general states and channels, they
have proven to be useful in a variate of tasks such as 
quantum communication~\cite{Grosshans2003}, quantum computation~\cite{PhysRevLett.82.1784}
and the study 
of quantum entanglement in simple~\cite{RevModPhys.77.513} and complicated
scenarios~\cite{PhysRevA.98.022335}.

Writing Gaussian channels in the \textit{position state representation} is
often of theoretical convenience, for instance for the calculation of position
correlation functions. Thus, an interesting way to proceed is to characterize
the possible functional forms of \gqc{} in such representation. First attempts
in this direction were given in Ref.~\cite{PazPhysRevLett}, but their
\textit{ansatz} is limited to only Gaussian functional forms (denoted simply by
\textit{Gaussian forms} or \gf{}). Going beyond such restrictive assumption, in
the present work we characterize another two possible forms that can arise
directly from the definition of Gaussian channel in the one-mode case. We thus
give a complete characterization of \gqc{} in position state representation,
and study the special case of \textit{singular Gaussian quantum channels}
(\sgqc{}), \ie{} the operations for which the inverse operation does not exist.
There are works that study similar representation problems.
For example, in Ref.~\cite{Moshinsky1971} the authors studied the position
state representation of Gaussian unitaries, suitable for close dynamics.
We study Gaussian channels, which include Gaussian unitaries. However, our
approach uses the quantum channel formalism (in the
position state representation), which includes open system dynamics and 
allows for singular forms in contrast to Ref.~\cite{Moshinsky1971}. 
In the same context of quantum channels, in Ref.~\cite{Ivan2011} the authors
studied the well known operator sum representation of Holevo's canonical forms,
which characterize the action of one-mode Gaussian channels upto Gaussian
unitaries~\cite{Holevo2007}. On the other hand, our characterization goes over
the full parametric space of one-mode Gaussian quantum channels, using the
position state representation.

One surprising result of the present study is that key properties of Gaussian
channels that contain delta-like factors {\it can not} be reproduced by their
smoothed form. This is particularly clear when considering binary properties
like unitarity, as we will explicitly show.

The paper is organized as follows. In section~\ref{sec:gqc} we discuss the
definition of \gqc{} and introduce functional forms beyond the \gf{} that emerge from
singularities in the coefficients that define a \gqc{} with \gf{}.
In section~\ref{sec:ptp}
we give a
\textit{black-box} characterization of such channels, using complete positivity and trace preserving conditions. In section~\ref{sec:singular_forms} we study functional forms that lead to \sgqc{} and derive their explicit form. Finally in section~\ref{sec:master_equations} we
derive conditions of existence of master equations and their explicit forms. We conclude in section~\ref{sec:conclusions}.
\section{Gaussian quantum channels}  
\label{sec:gqc}

Gaussian states are characterized completely
by first (mean) and second (correlations) moments encoded in the \textit{mean vector} $\vec d$
 and the \textit{covariance matrix} $\sigma$. Therefore, a Gaussian state $S$
can be denoted as $S=S\left(\sigma,\vec d\right)$, where for the one-mode case we have 
$$\sigma=\left( \begin{array}{cc}
\langle \hat q^2 \rangle-\langle \hat q \rangle^2 & \frac{1}{2}\langle \hat q \hat p+\hat p \hat q\rangle -\langle \hat q \rangle \langle \hat p \rangle \\ 
\frac{1}{2}\langle \hat q \hat p+\hat p \hat q \rangle -\langle \hat q \rangle \langle \hat p \rangle  & \langle \hat p^2 \rangle -\langle \hat p \rangle^2
\end{array} \right),$$ 
and $$\vec d=\left(\langle \hat q \rangle, \langle \hat p \rangle
\right)^\text{T}$$ with $\hat q$ and $\hat p$ denoting the standard position and momentum (quadrature) operators~\cite{Cerf}.

To start with, we recall the following definition~\cite{Reviewquantuminfo}:
\begin{definition}[Gaussian quantum channels]
A quantum channel is Gaussian (\gqc{}) if it transforms Gaussian
states into Gaussian states.
\label{def:gqc}
\end{definition}
This definition is strictly equivalent to the statement that any \gqc{}, say $\mcG$, can be written as
\begin{equation}
\mcG[\rho]=\tr_\text{E} \left[ U \left(\rho \otimes \rho_\text{E} \right) U^{\dagger} \right]
\label{def:def_2}
\end{equation}
where $U$ is a unitary transformation, acting on a combined global state obtained from enlarging the system with an  environment $\text{E}$,
that is generated by a quadratic bosonic Hamiltonian (\ie{} $U$ is a
\textit{Gaussian unitary})~\cite{Reviewquantuminfo}. The environmental initial state $\rho_E$ is a
Gaussian state and the trace is taken over the environmental
degrees of freedom.

Following definition~\ref{def:gqc}, a GQC is fully characterized by its action over Gaussian states, and this action is in turn defined by \textit{affine
transformations}~\cite{Reviewquantuminfo}. Specifically, $\mcG=\mcG\left( \mathbf{T}, \mathbf{N},
\vec \tau \right)$ is
given by a tuple $\left( \mathbf{T}, \mathbf{N},
\vec \tau \right)$ where $\mathbf{T}$ and $\mathbf{N}$ are $2\times 2$ real
matrices with $\mathbf{N}=\mathbf{N}^\text{T}$~\cite{Reviewquantuminfo} acting on Gaussian states according to  $\mcG\left( \mathbf{T},
\mathbf{N}, \vec \tau
\right)\left[S\left(\sigma,\vec d\right)\right]=S\left( \mathbf{T}\sigma
\mathbf{T}^{\text{T}} +\mathbf{N}, \mathbf{T} \vec d +\vec \tau \right)$. In
the particular case of closed systems we have $\mathbf{N}=\mathbf{0}$ and $\mathbf{T}$ is a
symplectic matrix,
which corresponds to linear canonical transformations. Their
unitary representations in the position state basis, together with 
their composition rules,  was already studied in
Ref.~\cite{Moshinsky1971}. The authors showed that Gaussian unitaries have
always a complex exponential form, while the superoperator
representation of Gaussian unitaries (\ie{} $\rho\mapsto U \rho U^\dagger{}$)
does not (see Lemma~\ref{thm:lemma}).

Let us note that although channels with Gaussian form trivially transform
Gaussian states into Gaussian states, the definition
goes beyond \gf{}.
 Introducing \textit{difference} and \textit{sum} coordinates with the notation used in Ref.~\cite{Gert}, $x=q_2-q_1$ and
$r=(q_1+q_2)/2$, such that $\rho(x,r)=\left.\left\langle
r-\frac{x}{2} \right. \right| \left. \hat \rho \left|  r+\frac{x}{2} \right.
\right\rangle$, a quantum channel 
\begin{equation}
\rho_f\left(x_f,r_f\right)
   =\int_{\mathbb{R}^2} dx_i dr_i J(x_f,x_i;r_f,r_i)\rho_i\left(r_i,x_i\right),
\label{eq:propagacion}
\end{equation}
maps an initial $\hat \rho_i$ into a final $\hat \rho_f$ state linearly through
the kernel $J(x_f,x_i;r_f,r_i)$. In order to see how a channel without \gf{}
can be constructed as a limiting case of a quantum channel with \gf{}, consider
the general parametrization of the latter as given in ~\cite{PazSupplementary}
\begin{eqnarray} \fl 
J_\text{G}(x_f,x_i;r_f,r_i)
   =\frac{b_3}{2 \pi} \exp\Big[\Imi\Big( 
     b_1 x_f r_f  +b_2x_f r_i  +b_3x_ir_f \nonumber &  +b_4x_ir_i +c_1x_f+c_2x_i \Big) 
 \\    
 &   -a_1 x_f^2-a_2x_fx_i-a_3x_i^2 \Big],
\label{eq:gf}
\end{eqnarray}
where all coefficients are real and no quadratic terms in $r_{i,f}$ are
allowed due to the hermiticity and trace preserving conditions. Now it is easy to see that if the coefficients of the quadratic form in the exponent of $J_\text{G}$ in \eref{eq:gf} depend on a parameter $\epsilon$ such that for $\epsilon \to 0$ they scale as $a_n \propto \epsilon^{-1}$ and $b_n \propto \epsilon^{-1/2}$, then 
\begin{equation}
\lim_{\epsilon \to 0}
J_\text{G}(x_f,x_i;r_f,r_i) = \mcN \delta(\alpha x_f-\beta x_i) e^{\Sigma'(x_f,x_i;r_f,r_i)},
\label{eq:typeII}
\end{equation}
where $\alpha$, $\beta \in \mathbb{R}$ and $\Sigma'(x_f,x_i;r_f,r_i)$ is a quadratic form that now admits quadratic terms in $r_{i,f}$. This is the first example of a \dgqc{}, namely a Gaussian quantum channel that contains Dirac-delta functions in its coordinate representation. This particular example is not only of academic interest. Physically, it can
be implemented by means of the ubiquitous quantum Brownian motion (QBM) model for
harmonic systems (damped harmonic oscillator)~\cite{Gert}. In such system
\dgqc{} occur at isolated points of time, defined in the limit of the
\textit{antisymmetric
position autocorrelation function} tending to zero.

Since the form of \eref{eq:typeII} admits quadratic terms in $r_{i,f}$, it
suggest that a form with an additional delta can exist and can be defined with  the usual limit of the Dirac delta using a Gaussian function. In order to avoid working with such limits, in this work we provide a \textit{black-box} characterization of general \gqc{}s without Gaussian form. In
particular we study channels that can arise when singularities on the
coefficients of Gaussian forms \gf{} occur, that lead immediately to singular Gaussian operations. We
characterize which forms in \dgqc{} lead to valid quantum channels, and under
which conditions singular operations lead to valid \textit{singular quantum
channels} (\sgqc{}). We will show that only two possible forms of \dgqc{} hold according to \textit{trace preserving} (\tp{}) and \textit{hermiticity preserving} (\hp{}) conditions. The channel of~\eref{eq:typeII} is one of these forms, as expected. Later on we will impose \textit{complete positivity} in order to have valid \gqc{}, \ie{} \textit{completely positive and trace preserving} (\cptp{}) Gaussian operations, going beyond previous characterizations of \gqc{} in position state representation~\cite{PazSupplementary}.

\section{Complete Positive and trace-preserving $\delta-$Gaussian operations}  
\label{sec:ptp}     
Let us introduce the ans\"atze for the possible forms of \gqc{} in the
position representation, to perform the black-box characterization. Following
\eref{def:def_2} and taking the continuous variable representation of
difference and sum coordinates, the trace becomes an integral over position
variables of the environment. Then we end up with a Fourier transform of a
multivariate Gaussian, having for one mode the
following structures: a Gaussian form \eref{eq:gf}, a Gaussian form multiplied
with one-dimensional delta or a  Gaussian form multiplied by a two-dimensional
delta. Thus, in order to start with the black-box characterization, we
shall propose the following general Gaussian operations with one and two
deltas, respectively
\begin{eqnarray}
J_\text{I}(x_f,r_f;x_i,r_i)&=
  \mcN_\text{I} \delta(\vec \alpha^{\text{T}}
     \vec v_f-\vec \beta^{\text{T}} \vec v_i)e^{\Sigma(x_f,x_i;r_f,r_i)}, 
 \label{eq:deltaop1} \\
J_\text{II}(x_f,r_f;x_i,r_i)&=\mcN_\text{II} \delta (\mathbf{A}\vec{v}_f-\mathbf{B} \vec v_i)e^{\Sigma(x_f,x_i;r_f,r_i)}, 
\label{eq:deltaop2}
\end{eqnarray} 
where the exponent reads
\begin{eqnarray}
\fl
\Sigma(x_f,x_i;r_f,r_i)=\Imi &\Big( b_1 x_f r_f
 +b_2  x_f r_i
 +b_3x_ir_f\nonumber 
+b_4x_ir_i+c_1x_f
+c_2x_i \Big)\nonumber\\
 &-a_1 x_f^2-a_2x_fx_i-a_3x_i^2 -e_1r_f^2-e_2r_fr_i-e_3r_i^2 -d_1r_f-d_2r_i,
\label{eq:quadratic_form}
\end{eqnarray} 
with 
$\vec{v}_{i(f)}=(r_{i(f)},x_{i(f)})$. 
$\mathbf{A}$ and $\mathbf{B}$
are $2\times 2$ matrices, $\vec \alpha$ and $\vec \beta$ are two-dimensional
vectors, and $\mcN_\text{I,II}$ are normalization constants. 
They provide, together with \eref{eq:gf} all possible ans\"atze for \gqc{}.
Note that the coefficients in the exponential of every form must be finite,
otherwise the functional form can be modified.

Let us study now \cptp{} conditions, since \textit{complete positivity} implies
\textit{positivity} and in turn it implies \textit{hermiticity preserving}
(\hp{}). For sum and difference coordinates \hp{} reads
\begin{equation}
J(-x_f,r_f;-x_i,r_i)=J(x_f,r_f;x_i,r_i)^{*}.
\label{eq:orale}
\end{equation}
Following this equation, it is easy to note that the coefficients $a_n$, $b_n$,
$c_n$, $e_n$ and $d_n$ must be real, as well the entries of matrices (and
vectors) $\mathbf{A}$, $\mathbf{B}$, $\vec \alpha$, $\vec \beta$. The 
delta function in \eref{eq:deltaop1} is reduced to  i) $\delta(\alpha x_f-\beta x_i) $ 
or 
ii) $\delta(\alpha r_f-\beta r_i) $ when fulfilling condition (\ref{eq:orale}). For the case of \eref{eq:deltaop2},
the two-dimensional delta is reduced to iii) $
\delta(\gamma r_f-\eta r_i)\delta(\alpha x_f-\beta x_i)$.
Let us now analyze the trace preserving condition (\tp{}), which for continuous variable systems reads
\begin{equation}
\int_{\mathbb{R}} dr_fJ(x_f=0,r_f;x_i,r_i)=\delta (x_i).
\label{eq:tp}
\end{equation} 
This condition immediately discards ii) from the above combinations of deltas,
thus we end up with cases i) and iii). For case i) \tp{} reads
\begin{equation}
\mcN_\text{I}\int dr_f\delta(-\beta x_i)e^{\Sigma}
    =\frac{\mcN_\text{I}}{|\beta|}\sqrt{\frac{\pi}{e_1}}\delta(x_i)
	e^{\left(\frac{e_2^2}{4e_1}-e_3\right)r_i^2},
\end{equation}
thus the relation between the coefficients assumes the form
\begin{equation}
\frac{e_2^2}{4e_1}-e_3=0, d_1=0,d_2=0,
\label{eq:trace_preserving_form1}
\end{equation}
and the normalization constant $\mcN_\text{I}=|\beta |\sqrt{\frac{e_1}{\pi}}$ with
$\beta\neq 0$ and $e_1>0$.  For case iii) the trace-preserving condition reads	
\begin{eqnarray*} \fl
\mcN_\text{II}
\int dr_f \delta(\gamma r_f-\eta r_i)\delta(-\beta x_i)e^{\Sigma}
	=\frac{\mcN_\text{II}}{|\beta \gamma|}\delta(x_i) e^{-\left(e_1(\frac{\eta}{\gamma})^2+e_2\frac{\eta}{\gamma}+e_3\right)r_i^2-\left(d_1\frac{\eta}{\gamma}+d_2\right) r_i}.
\end{eqnarray*}
Thus, the following relation between $e_n$ and $d_n$ coefficients must be fulfilled:
\begin{equation}
e_1\Bigg(\frac{\eta}{\gamma}\Bigg)^2+e_2\frac{\eta}{\gamma}+e_3=0, \ \ d_1\frac{\eta}{\gamma}+d_2=0,
\label{eq:restricII}
\end{equation}
with $\gamma,\beta \neq 0$ and $\mcN_\text{II}=|\beta \gamma|$. In the particular case of $\eta=0$, \eref{eq:restricII} is reduced to $e_3=d_2=0$. 
As expected from the analysis of limits above, we showed that \dgqc{}s admit quadratic terms in $r_{i,j}$. 

Up to this point we have \textit{hermitian and trace preserving Gaussian
operations}; to derive the remaining CPTP conditions, it is useful to write
Wigner's function and Wigner's characteristic function, which we now derive. 
The representation of the Wigner's characteristic function reads
\begin{equation}
\chi(\vec k)=\exp\left[ -\frac{1}{2}\vec k^{\text{T}}\left( \Omega \sigma \Omega^{\text{T}}\right)\vec k-\Imi\left( \Omega \langle \hat x \rangle \right)^{\text{T}} \vec k \right]
\label{eq:cha}
\end{equation}
and its relation with Wigner's function
\begin{eqnarray} 
\fl
W(\mathbf{x})&=\frac{1}{(2\pi)^2}\int_{\mathbb{R}^{2}} d\vec k e^{-\Imi \vec{x}^{\text{T}}\Omega \vec k}\chi\left(\vec k\right)
=\frac{1}{2\pi}\int_{\mathbb{R}} dx e^{\Imi p x} \left.\left\langle  r-\frac{x}{2} \right. \right| \left. \hat \rho \left|  r+\frac{x}{2} \right. \right\rangle.
\label{eq:wigners}
\end{eqnarray}
where $\vec k=\left(k_1,k_2 \right)^{\text{T}}$, $\vec
x=\left(r,p\right)^{\text{T}}$ and $\hbar=1$ (we are using natural units).
Using the previous equations to construct Wigner and Wigner's
characteristic functions of the initial and final states, and substituting them
in \eref{eq:propagacion}, it is straightforward to get the propagator in the Wigner's
characteristic function representation
\begin{eqnarray} 
\tilde J \left( \vec k_f,\vec k_i \right)=\int_{\mathbb{R}^6} d\Gamma K\left(\vec l\right) J(\vec v_f,\vec v_i),
\label{eq:changeofrepresentation}
\end{eqnarray} 
where the transformation kernel reads
$$K\left(\vec{l}\right)= \frac{1}{(2\pi)^3}e^{  \left[ \Imi \left( k_2^f r_f-k_1^f p_f -k_2^i r_i+k_1^i p_i -p_i x_i +p_f x_f \right)  \right]  },$$ 
with 
\begin{eqnarray*} 
d\Gamma&=d p_f d p_i dx_f dx_i dr_f dr_i\text{ and }\\
\vec l&=\left(p_f, p_i, x_f, x_i, r_f,r_i \right)^{\text{T}}.
\end{eqnarray*} 
By elementary integration of ~\eref{eq:changeofrepresentation} one can
show that for both cases
\begin{equation}
\tilde J_{\text{I,III}}\left( \vec k_f, \vec k_i \right)= \delta \left ( k_1^i -\frac{\alpha}{\beta} k_1^f  \right) \delta \left( k_2^i -\vec\phi^{\text{T}}_{\text{I,III}} \vec k_f  \right) e^{P_\text{I,III}(\vec k_f)},
\label{eq:formcharI}
\end{equation}
where $P_{\text{I,III}}(\vec k_f)=\sum_{i,j=1}^2 P^{(\text{I,III})}_{ij} k^f_i
k^f_j+\sum_{i=1}^2 P^{(\text{I,III})}_{0i}k^f_i$ with
$P^{(\text{I,III})}_{ij}=P^{(\text{I,III})}_{ji}$.
For case i) we obtain
\begin{eqnarray} 
P^\text{(I)}_{11}&=-\left(\left(\frac{\alpha}{\beta}\right)^2\left(a_3+\frac{b_3^2}{4 e_1}\right)+\frac{\alpha}{\beta} \left(a_2+\frac{1}{2} \frac{b_1 b_3}{e_1}\right)+a_1+\frac{b_1^2}{4 e_1}\right),\nonumber\\
P^\text{(I)}_{12}&=-\left(\frac{\alpha}{\beta}\frac{b_3}{2 e_1}+\frac{b_1}{2 e_1}\right),\nonumber\\
P^\text{(I)}_{22}&=-\frac{1}{4 e_1}.
\label{eq:deltaI_Ps}
\end{eqnarray} 
For case iii) we have
\begin{eqnarray} 
P^\text{(III)}_{11}&= -\left(\left(\frac{\alpha}{\beta}\right)^2a_3 +\frac{\alpha}{\beta} a_2 +a_1\right),\nonumber\\
P^\text{(III)}_{12}&=P^\text{(III)}_{22}=0,
\label{eq:deltaII_Ps}
\end{eqnarray} 
And for both cases we have $P^{(\text{I,III})}_{01}=\Imi \left( \frac{\alpha}{\beta}c_2 +c_1\right)$ and $P^{(\text{I,III})}_{02}=0$.
Vectors $\vec \phi$ are given by
\begin{eqnarray} 
\vec\phi_{\text{I}}&=\left(\frac{\alpha}{\beta}\left( b_4-\frac{b_3 e_2}{2 e_1}\right)-\frac{b_1 e_2}{2 e_1}+b_2,-\frac{e_2}{2 e_1}\right)^\text{T},\nonumber\\
\vec\phi_{\text{III}}&=\left( \frac{\alpha}{\beta}\frac{\eta}{\gamma}b_3+\frac{\alpha}{\beta}b_4 +\frac{\eta}{\gamma}b_1 +b_2,\frac{\eta}{\gamma} \right)^\text{T}.
\label{eq:phis}
\end{eqnarray} 

We are now in position to write explicitly the conditions for complete positivity. 
Having a Gaussian operation
characterized by
$\left( \mathbf{T}, \mathbf{N}, \vec \tau \right)$, 
the CP condition can be expressed in terms of the matrix
\begin{equation}
\mathbf{C}=\mathbf{N}+\Imi\Omega -\Imi\mathbf{T}\Omega \mathbf{T}^{\text{T}},
\label{eq:ccp}
\end{equation}
where $\Omega=\left( \begin{array}{cc}
0 & 1 \\ 
-1 & 0
\end{array}  \right)$ is the symplectic matrix. An operation $\mcG\left(
\mathbf{T}, \mathbf{N}, \vec \tau \right)$ is CP if and only if $\mathbf{C}\geq
0$~\cite{cptp,Reviewquantuminfo}. 
Applying the propagator on a test characteristic function, \eref{eq:cha}, it is
easy to compute the corresponding tuples. For both cases we get $\left(
\mathbf{T}_\text{I,III},\mathbf{N}_\text{I,III}, \vec \tau_\text{I,III}
\right)$:
\begin{eqnarray} 
\mathbf{N}_\text{I,III}&=2\left( \begin{array}{cc}
-P_{22} & P_{12} \\ 
P_{12} & -P_{11}
\end{array}  \right),\nonumber\\
\vec \tau_{\text{I,III}}&=\left(0,\Imi{} P^{(\text{I,III})}_{01} \right)^\text{T},
\label{eq:nsandtaus}
\end{eqnarray} 
while for case i) matrix $\mathbf{T}$ is given by
\begin{equation}
\mathbf{T}_\text{I}=\left(
\begin{array}{cc}
 \frac{e_2}{2 e_1} & 0 \\
 \vec \phi_\text{I,1} & -\frac{\alpha }{\beta } \\
\end{array}
\right),
\label{eq:tupleI}
\end{equation}
where $\vec \phi_\text{I,1}$ denotes the first component of vector $\vec \phi_\text{I}$, see \eref{eq:phis}.
The complete positive condition is given by the inequalities raised from the eigenvalues of matrix~\eref{eq:ccp}
\begin{eqnarray} 
\fl
\pm&\frac{\sqrt{\alpha ^2 e_2^2+4 \alpha  \beta  e_2 e_1+4 \beta ^2 e_1^2 \left(4 {P_{12}^{(\text{I})}}^2+\left(P_{11}^{(\text{I})}-P_{22}^{(\text{I})}\right)^2+1\right)}}{2 \beta  e_1}
&-\left(P_{11}^{(\text{I})}+P_{22}^{(\text{I})}\right) \geq 0.
\label{eq:CPI}
\end{eqnarray} 
For case iii) matrix $\mathbf{T}$ is
\begin{equation}
\mathbf{T}_\text{III} = \left(
\begin{array}{cc}
 -\frac{\eta }{\gamma } & 0 \\
 \vec \phi_\text{III,1}  & -\frac{\alpha }{\beta } \\
\end{array}
\right),
\label{eq:tupleII}
\end{equation}
and complete positivity conditions read
\begin{equation}
\pm \frac{\sqrt{(\beta  \gamma -\alpha  \eta )^2+\beta ^2 \gamma ^2 {P_{11}^{(\text{III})}}^2}}{\beta  \gamma }-P_{11}^{(\text{III})}\geq 0.
\label{eq:cpII}
\end{equation}
Note that in both cases complete positivity conditions do not depend on $\vec \phi$.
These results will be discussed in the next section for the singular case.
\section{Allowed singular forms}  
\label{sec:singular_forms}
There are two classes of Gaussian singular
channels. Since the inverse of a Gaussian channel $\mcG\left(
\mathbf{T},\mathbf{N}, \vec \tau \right)$ is $\mcG\left( \mathbf{T}^{-1},
-\mathbf{T}^{-1}\mathbf{N} \mathbf{T}^{-T}, -\mathbf{T}^{-1} \vec \tau
\right)$, its existence rests on the invertibility of $\mathbf{T}$. Thus,
studying the rank of the latter we are able to explore singular forms. 
We are going to use the classification of one-mode channels developed by 
Holevo~\cite{Holevo2007}.

For singular channels there are two classes
characterized by its \textit{canonical form}~\cite{Holevo2008}, \ie{} any 
channel can be obtained by applying Gaussian unitaries before and
after the canonical form. The class called ``$A_1$'' correspond to singular
channels with $\text{Rank} \left(\mathbf{T}\right)=0$ and coincide with the
family of \textit{total depolarizing channels}. The class ``$A_2$'' is
characterized by $\text{Rank}\left(\mathbf{T}\right)=1$. Both classes are
entanglement-breaking~\cite{Holevo2008}.

Before analysing the functional forms constructed in this work, let us study
channels with \gf{}. The tuple of the affine transformation, corresponding to
the propagator $J_\text{G}$, \eref{eq:gf}, were introduced in
Ref.~\cite{PazSupplementary} up to some typos. Our calculation
for this tuple is
\begin{eqnarray} 
\mathbf{T}_\text{G}&=\left(
\begin{array}{cc}
 -\frac{b_4}{b_3} & \frac{1}{b_3} \\
 \frac{b_1 b_4}{b_3}-b_2 & -\frac{b_1}{b_3} \\
\end{array}
\right),\nonumber \\
\mathbf{N}_\text{G}&=\left(
\begin{array}{cc}
 \frac{2 a_3}{b_3^2} & \frac{a_2}{b_3}-\frac{2 a_3 b_1}{b_3^2} \\
 \frac{a_2}{b_3}-\frac{2 a_3 b_1}{b_3^2} & -2 \left(-\frac{a_3 b_1^2}{b_3^2}+\frac{a_2 b_1}{b_3}-a_1\right) \\
\end{array}
\right),\nonumber \\
\vec \tau_\text{G}&= \left( -\frac{c_2}{b_3},\frac{b_1 c_2}{b_3}-c_1 \right)^\text{T}.
\label{eq:singular_tupleII}
\end{eqnarray} 
It is straightforward to check that for $b_2=0$, $\mathbf{T}_\text{G}$ is
singular with $\text{Rank}\left(\mathbf{T}_\text{G}\right)=1$, \ie{} it belongs
to class $A_2$. Due to the full support of Gaussian functions, it was
surprising that Gaussian channels with \gf{} have singular limit. In this case
the singular behavior arises from the lack of a Fourier factor for $x_f r_i$.
This is the only singular case for \gf{}.

It is instructive at this point to analyze the physical meaning of
the transition from non-singular to singular Gaussian channels in the
context of QBM~\cite{Gert}. There, a central harmonic
oscillator is linearly coupled through its position with a bath of harmonic oscillators
initially at thermal equilibrium. The coupling of the central system
to the bath causes dissipation of energy and decoherence in the position basis.
At any time $t$ the quantum channel of QBM has the Gaussian form given by
\eref{eq:gf}, with time-depending coefficients.  Excluded here are particular
times with divergent coefficients that need to be treated differently, as
already mentioned in section~\ref{sec:ptp}. 

For a specific form of the bath spectral density all these coefficients are
calculated in Ref.~\cite{Gert}, and in particular the coefficient $b_2$ is found to
be proportional to $e^{-\gamma t/2}\sin(\eta t)$, where $\gamma$ is the damping
parameter characterizing the interaction strength and $\eta$ is a coefficient
related with the natural frequency of the central oscillator. It is clear that for any finite $\gamma$,
$b_2 \to 0$  when $t \to \infty$, and then for long times the matrix
$\mathbf{T}_\text{G}$ becomes singular. It is precisely at $t \to \infty$ when
the particle reaches the state of equilibrium in the QBM model (the reduced
density matrix of the particle reaches a stationary state). Moreover, it can be
shown that the quantum channel in the QBM model acquires the form $J_\text{I}$,
with $b_2=b_1=c_1=c_2=0$, which actually defines a singular channel, as shown
above. Finally, it is easy to see that the determinant of the
$\mathbf{T}_\text{G}$ matrix in QBM is proportional to $\eta^ 2e^ {-\gamma t}$,
and thus $\mathbf{T}_\text{G}$ becomes singular only when $t \to \infty$. This
shows the deep relation between dynamical irreversibility and singular channels
in this context.

Now we analyze functional forms derived in~\sref{sec:ptp}. The complete
positivity conditions of the form $\tilde J_\text{III}$, presented in
\eref{eq:cpII}, have no solution for $\alpha\to 0$ and/or $\gamma \to 0$, thus
this form cannot lead to singular channels. This is not the case for $\tilde
J_\text{I}$, \eref{eq:formcharI}, which leads to singular operations belonging
to class $A_2$ for $\alpha e_2=0$,  and to class $A_1$ for $e_2 =\alpha=b_2= 0$.
For the latter, the complete
positivity conditions read 
\begin{equation}
e_1\leq a_1.
\label{eq:reducedCP}
\end{equation}

By using an initial state characterized by $\sigma_i$ and $\vec d_i$ we
can compute the explicit dependence of the final states on the initial
parameters. For channels belonging to class $A_2$ [see
\eref{eq:singular_tupleII} with $b_2=0$ and \eref{eq:tupleI} with $e_2
\alpha=0$] the final state only depends in one combination of the
components of $\sigma_i$, and in one combination of the components of $\vec
d_i$, \ie{} $ \sum_{mn} l_{mn} \left(\sigma_i \right)_{mn} $ and $\sum_m n_m
\left( \vec d_i \right)_m$, respectively, where $l_{mn}$ and $n_m$ depend on
the channel parameters. See the appendix for the explicit formulas
and~\fref{fig:1} (right panel) for an schematic description of the final states. From such
combinations it is obvious that we cannot solve for the initial state
parameters given a final state as expected; this is because the parametric
space dimension is reduced from $5$ to $2$. The channel belonging to $A_1$ [see
\eref{eq:tupleI} with $e_2 =\alpha=b_2= 0$ and \eref{eq:reducedCP}] maps every
initial state to a single one characterized by $\sigma_f=\mathbf{N}$ and $\vec d_f=
\left(0,-c_1 \right)^\text{T}$, see \fref{fig:1}  (left panel) for a schematic description.

According to our ans\"atze [see eqs.~(\ref{eq:deltaop1}) and
(\ref{eq:deltaop2})], we conclude that one-mode \sgqc{} can only have the functional forms
given in \eref{eq:gf} and \eref{eq:deltaop1}.
This is one of the central results of our work and can be stated as:
\begin{theorem}[One-mode singular Gaussian channels]
A one-mode Gaussian quantum channel is singular if and only if it
has one of the following functional forms in the position space representation
\begin{enumerate}
\item $\frac{b_3}{2 \pi} \exp\Big[ \Imi{}\Big( b_1 x_f r_f+b_3x_ir_f 
+b_4x_ir_i+c_1x_f+c_2x_i \Big) -a_1 x_f^2-a_2x_fx_i-a_3x_i^2 \Big],$
\item $|\beta|\sqrt{e_1/\pi}\delta(\alpha x_f-\beta x_i) 
\exp\Big[-a_2 x_f x_i-a_1 x_f^2-a_3 x_i^2$\\
$+\Imi{} \Big(b_2 x_f r_i +b_3 r_f x_i+b_1 r_f x_f+b_4 r_i x_i+c_1 x_f+c_2 x_i\Big)$\\
$-e_1 r_f^2 -e_2 r_f  
               r_i-\frac{e_2^2 r_i^2}{4 e_1}\Big]$, with $e_2 \alpha=0$.
\end{enumerate}
\label{thm:teorema}
\end{theorem}
\begin{corollary}[Singular classes]
A one-mode singular Gaussian channel belongs to class $A_1$ if and only if its position representation has the following form:
$$\sqrt{e_1/\pi}\delta(x_i) \exp\Big[-a_1 x_f^2+\Imi{} \Big(b_2 x_f r_i 
         +b_1 r_f x_f+c_1 x_f\Big)-e_1 r_f^2\Big].$$
Otherwise the channel belongs to class $A_2$.
\label{thm:corollario}
\end{corollary}

As a particular example of Theorem~\ref{thm:teorema}, choosing the case (ii) with $e_2=2 e_1=(2 \bar{n}+1)^{-1}$, $\alpha=0=b_1=b_2=0$, $a_1=(2\bar{n}+1)/2$ one gets the Holevo's canonical form of class $A_2$, see Ref.~\cite{Reviewquantuminfo}.
 The parameter $\bar{n}$ is the average of the number of excitation in the mode. Notice that the canonical form cannot be obtained using the Gaussian functional form, while elements in $A_2$ enjoy both (i) and (ii) forms of Theorem~\ref{thm:teorema} (see Table~\ref{tab:concats}). Additionally the operator sum representation of this canonical form is given in Ref.~\cite{Ivan2011}.

Since channels on each class are connected each other by unitary
conjugations~\cite{Holevo2007}, a consequence of the theorem and the subsequent
corollary is that the set of allowed forms must remain invariant under unitary
conjugations. 
To show this we must know the possible functional forms of Gaussian unitaries. They are given by following lemma for one mode
\begin{lemma}[One-mode Gaussian unitaries]
Gaussian unitaries have only \gf{} or the one given by \eref{eq:deltaop2}.
\label{thm:lemma}
\end{lemma}
\begin{proof}
Recalling that for a unitary \gqc{},  $\mathbf{T}$  must be symplectic
($\mathbf{T}\Omega \mathbf{T}^\text{T}=\Omega$) and $\mathbf{N}=\mathbf{0}$. However,
an inspection to \eref{eq:deltaI_Ps} lead us to note that
$\mathbf{N}\neq\mathbf{0}$ unless $e_1$ diverges. Thus, Gaussian unitaries
cannot have the form $J_\text{I}$ [see~\eref{eq:deltaop1}]. An inspection of
matrices $\mathbf{T}$ and $\mathbf{N}$ of \gqc{} with \gf{} [see
\eref{eq:singular_tupleII}] and the ones for $J_\text{II}$ [see equations
(\ref{eq:deltaII_Ps}) and (\ref{eq:tupleII})] lead us to note the following 
two observations: (i) in
both cases we have
$\mathbf{N}=0$ for $a_n=0 \ \ \forall n$; (ii) the matrix $\mathbf{T}$ is
symplectic for \gf{} when $b_2=b_3$, and when $\alpha  \eta =\beta  \gamma $ for $J_\text{II}$. In particular the identity map has the last form. This completes
the proof.
\end{proof}

\newcommand{\deltaAa}{$\delta_{\text{A}_2}^\alpha$}
\newcommand{\deltaAe}{$\delta_{\text{A}_2}^{e_2}$}
\newcommand{\deltaAae}{$\delta_{\text{A}_2}^{\alpha,e_2}$}
\newcommand{\deltaA}{$\delta_{\text{A}_1}$}
\newcommand{\gaussA}{$\mcG_{\text{A}_2}$}
\newcommand{\deltaU}{$\delta_{\mcU}$}
\newcommand{\gaussU}{$\mcG_{\mcU}$}
One can now compute the concatenations of the \sgqc{}s with Gaussian unitaries. This can be done straightforward using the well known formulas for Gaussian integrals and the Fourier transform of the Dirac delta.
Given that the calculation is elementary, and for sake of brevity, we present
only the resulting forms of each concatenation. To show this compactly we
introduce the following abbreviations: Singular channels belonging to class
$A_2$ with form $J_\text{I}$ and with $\alpha=0$, $e_2=0$ and $\alpha=e_2=0$,
will be denoted as \deltaAa{}, \deltaAe{} and \deltaAae{}, respectively;
singular channels belonging to the same class but with \gf{} will be denoted as
\gaussA{}; channels belonging to class $A_1$ will be denoted as \deltaA{};
finally Gaussian unitaries with \gf{} will be denoted as \gaussU{} and the ones
with form $J_\text{II}$ as \deltaU{}. Writing the concatenation of two channels
in the position representation as
\begin{eqnarray} \fl 
\qquad
\qquad
J^{(\text{f})}(x_f,r_f;x_i,r_i)=
\int_{\mathbb{R}^2} dx' dr' J^\text{(1)}\left(x_f, r_f; x',r' \right)J^{(2)}\left( x',r';x_i, r_i\right),
\label{eq:concat}
\end{eqnarray} 
the resulting functional forms for $J^{(\text{f})}$ are given in 
table~\ref{tab:concats}. As expected, the table shows that the integral has
only the forms stated by our theorem. Additionally it shows the cases when
unitaries change the functional form of class $A_2$, while for class $A_1$,
$J^{(\text{f})}$ has always the unique form enunciated by the corollary. This table also constitutes another proof for our theorem and the corollary.

The central results of this section are the following. The three functional
forms existing for one-mode Gaussian channels in position state representation
allow specific types of channels. 
We proved that only channels with \gf{} or $J_\text{II}$ can be
unitary, whereas channels with \gf{} or $J_\text{I}$ can be singular.
To show the relevance of this result, we introduce the following example
in which the
transition from one functional form to another determines the change from
unitarity to non-unitarity.  
\begin{example}[Unitary transition]
Let
$\delta(r_f- r_i \eta/\gamma)\delta(x_f \gamma /\eta -x_i) \exp(\Sigma)$
be a family of unitary channels with functional form $J_\text{II}$,
fulfilling eqs.~(\ref{eq:restricII}) and
(\ref{eq:cpII}) with $a_n=0$, $\forall n$. Substituting
$\delta(r_f-\eta/\gamma) \to 
\exp[-(r_f-r_i \eta/\gamma)^2/\epsilon^2]/(\epsilon \sqrt{\pi})$
with $\epsilon>0$, we recover the mentioned family of channels in the limit of
$\epsilon\to 0$. However, for finite $\epsilon$ we have a family of valid channels
with form $J_\text{I}$, see \eref{eq:deltaop1}. It is easy to check that
eqs.~(\ref{eq:trace_preserving_form1}) and (\ref{eq:CPI}) are fulfilled, the
latter for $\gamma\neq 0$ and $\beta \neq 0$. Since the form $J_\text{I}$
cannot describe unitary channels, the transition $J_\text{I}\leftrightarrow
J_\text{II}$ coincides with the transition non-unitary/unitary.
\end{example}
This example shows how the regularization of delta like factors in quantum channels leads to lose key properties, such as unitarity. This motivates investigating if other quantum
information properties are lost when performing regularizations.

\begin{table}[h] 
\centering
{\large
\begin{tabular}{|c|c|c|}
\hline 
$J^{(1)}$ & $J^{(2)}$ & $J^{(\text{f})}$ \\ 
\hline 
\deltaAa{} & \gaussU{} & \gaussA{} \\ 
\hline 
\gaussU{} & \deltaAa{} & \deltaAa{} \\ 
\hline 
\deltaAa{} & \deltaU{} & \deltaAa{} \\ 
\hline 
\deltaU & \deltaAa{} & \deltaAa{} \\ 
\hline 
\deltaAe{} & \gaussU{} & \deltaAe{} \\ 
\hline 
\gaussU{} & \deltaAe{} & \gaussA{} \\ 
\hline 
\deltaAe{} & \deltaU{} & \deltaAe{} \\ 
\hline 
\deltaU{} & \deltaAe{} & \deltaAe{} \\ 
\hline 
\gaussU{} & \deltaAae{} & \deltaAae{} \\ 
\hline 
\deltaAae{} & \gaussU{} & \gaussA{} \\ 
\hline
\deltaU{} & \deltaAae{} & \deltaAae{} \\ 
\hline 
\deltaAae{} & \deltaU{} & \deltaAae{} \\ 
\hline
\deltaU{},\gaussU{} & \deltaA{} & \deltaA{} \\ 
\hline 
\deltaA{} & \deltaU{},\gaussU{} & \deltaA{} \\ 
\hline 
\end{tabular}
}
\caption{The first and second columns show the functional forms of $J^{(1)}$ and $J^{(2)}$, respectively. The last column shows the resulting form of the concatenation of them [see~\eref{eq:concat}]. See main text for symbol coding. \label{tab:concats}}
\end{table}  

%

\begin{figure} 
\centering
\includegraphics[width=\columnwidth]{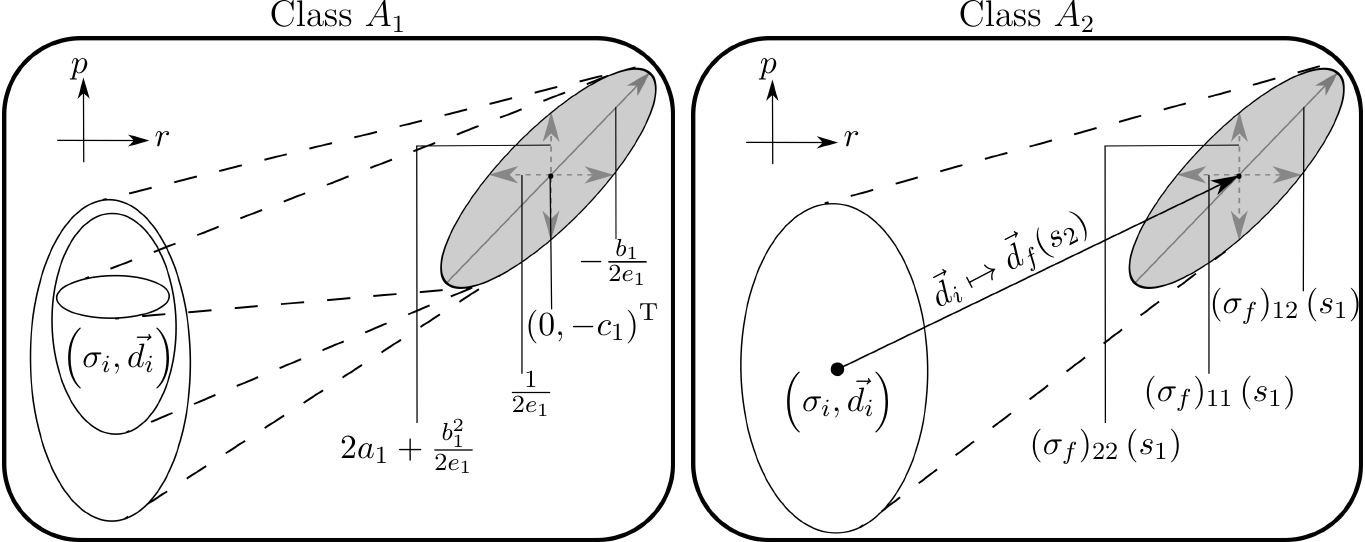}
\caption{
Schematic pictures of the channels belonging to classes $A_1$ and $A_2$ (right and left panels, respectively), acting on pictorial Wigner's functions of Gaussian states (represented with ellipses).
The coordinate system corresponds to the position variable $r$ and its conjugate momentum, $p$.
The figure shows how every channel in class $A_1$ maps every initial quantum state, in particular \gs{}s characterized by $\left(\sigma_i, \vec d_i \right)$, to a Gaussian state that depends only on the channel parameters. The values of the first and second moments of the final Gaussian state are indicated by a gray ellipse.
Similarly, for class $A_2$, we indicate the form of the final moments for initial Gaussian states. In this case they depend on two combinations of the initial parameters, $s_1$ and $s_2$, whose explicit formulas are given in the appendix, together with the form of the final moments.
\label{fig:1}
}
\end{figure} 

\section{Existence of master equations}  
\label{sec:master_equations}
In this section we show the conditions under which master equations, 
associated with the channels derived in~\sref{sec:ptp}, exist.
To be more precise, we prove that, under a simple condition,
the functional forms derived above
parametrize channels belonging to one-parameter differentiable families of
\gqc{}s. 
 The latter, together the result given in Ref.~\cite{PazSupplementary}, implies that channels with any functional form can be found in the vicinity of the identity map.
As a first step to derive the condition of existence of a differentiable family, we let the coefficients of forms presented in
eqs.~(\ref{eq:deltaop1}) and (\ref{eq:deltaop2}) to depend on time. Later we
derive the conditions under which they bring
any quantum state $\rho(x,r;t)$ to $\rho(x,r;t+\epsilon)$ (with $\epsilon>0$
and $t\in [0,\infty)$) smoothly, while holding the specific functional form of
the channel, \ie{} 
\begin{equation}
\rho(x,r;t+\epsilon)=\rho(x,r;t)+\epsilon \mcL_t
\left[\rho(x,r;t)\right]+\mcO(\epsilon^2),
\end{equation}
where both $\rho(x,r;t)$ and $\rho(x,r;t+\epsilon)$ are propagated from $t=0$ with channels either with the form $J_\text{I}$ or $J_\text{II}$, and $\mcL_t$ is a bounded superoperator in the state subspace.
This is basically the problem of \textit{the existence of a master equation}
\begin{equation}
\partial_t \rho(x,r;t)=\mathcal{L}_t\left[\rho(x,r;t)\right],
\label{eq:master_equation}
\end{equation} 
for such functional forms. Thus, the problem is reduced to prove the existence
of the linear generator $\mcL_t$, also known as \textit{Liouvillian}.

To do this we use an ansatz proposed in Ref.~\cite{Karrlein1997} to investigate the
existence and derive the master equation for \gf{}s,
\begin{eqnarray} \fl 
\qquad
\mathcal{L}=\mathcal{L}_c(t)+(\partial_x,\partial_r) 
\mathbf{X}(t) \left( \begin{array}{c} \partial_x \\ \partial_r \end{array} \right) 
+(x,r) 
\mathbf{Y}(t) 
\left( \begin{array}{c} \partial_x \\ \partial_r \end{array} \right)
+(x,r) 
\mathbf{Z}(t) \left( \begin{array}{c} x \\ r \end{array}
\right)
\label{eq:ansatz}
\end{eqnarray} 
where $\mathcal{L}_c(t)$ is a complex function and 
\begin{equation}
\mathbf{X}(t) =
\left( \begin{array}{cc} 
X_{xx}(t) & X_{xr}(t) \\
X_{rx}(t) & X_{rr}(t)
\end{array} \right)
\label{eq:defX}
\end{equation}
is a complex matrix as well as 
$\mathbf{Y}(t)$ and $\mathbf{Z}(t)$, whose entries are defined in a similar
way as in \eref{eq:defX}. Note that
$\mathbf{X}(t)$ and $\mathbf{Z}(t)$ can always be chosen symmetric, \ie{}
$X_{xr}=X_{rs}$ and $Z_{xr}=Z_{rx}$. Thus, we must determine $11$ time-dependent
functions from~\eref{eq:ansatz}. This ansatz is
also appropriate to study the functional forms introduced in this work, given
that the left hand side of \eref{eq:master_equation} only involves quadratic
polynomials in $x$,
$r$, $\partial/\partial x$ and $\partial/\partial r$, as in the \gf{} case.

Notice that singular channels do not admit a master equations since its
existence implies that channels with the functional form involved can be found
arbitrarily close from the identity channel. This is not possible for singular
channels due to the continuity of the determinant of the matrix $\mathbf{T}$.

For the non-singular cases presented in eqs.~(\ref{eq:deltaop1}) and (\ref{eq:deltaop2}), the condition for the
existence of a master equation is obtained as follows. (i) Substitute
the ansatz of \eref{eq:ansatz} in the right hand side of the
\eref{eq:master_equation}. (ii) Define $\rho(x,r;t)$ using \eref{eq:propagacion}, given an initial condition
$\rho(x,r;0)$, for each functional form $J_\text{I,II}$. (iii) Take
$\rho_f(x_f,r_f)\to\rho(x,r;t)$ and $\rho_i(x_i,r_i)\to \rho(x,r;0)$. Finally, (iv) compare both sides of
\eref{eq:master_equation}. Defining $A(t)=\alpha(t)/\beta(t)$ and
$B(t)=\gamma(t)/\eta(t)$, the conclusion is that for both $J_\text{I}$ and
$J_\text{II}$, a master equation exists if and only if $c_1(t) \dot A (t)=A(t)(\dot c_1(t)+A(t) c_2(t))$, this can be easily simplified, by adding $A(t)\dot A(t) c_2(t)$ in both sides the equation and integrating respect to $t$, to
\begin{equation}
c(t)\propto A(t),
\label{eq:master_exist}
\end{equation}
where $c(t)=
c_1(t)+A(t)c_2(t)$. Additionally, for the form $J_I$ the solutions for the
matrices $\mathbf{X}(t)$, $\mathbf{Y}(t)$ and $\mathbf{Z}(t)$ are given by 
\begin{eqnarray} 
X_{xx}=X_{xr}=Y_{rx}=Z_{rr}=0,\nonumber \\
Y_{xx}=\frac{\dot{A}}{A},\nonumber \\
\mathcal{L}_c=Y_{rr}=\frac{\dot{e}_1}{e_1}-\frac{\dot{e}_2}{e_2},\nonumber \\
X_{rr}=\frac{\dot{e}_1}{4e_1^2}-\frac{\dot{e}_2}{2e_1e_2},\nonumber \\
Y_{xr}=\Imi \left(\frac{\lambda_1\dot{e}_2}{e_1e_2}+\frac{\lambda_2\dot{A}}{e_2A}-\frac{\lambda_1\dot{e}_1}{2e_1^2}-\frac{\dot{\lambda}_2}{e_2}\right),\nonumber \\
Z_{xx}=\frac{\lambda_1^2}{2}\Bigg(\frac{\dot{e}_2}{e_1e_2}-\frac{\dot{e}_1}{2e_1^2}\Bigg)+\frac{\lambda_1}{e_2}\Bigg(\frac{\lambda_2\dot{A}}{A}-\dot{\lambda}_2\Bigg)+2\lambda_3\frac{\dot{A}}{A}-\dot{\lambda}_3, \nonumber \\
Z_{xr}=\Imi\left(\frac{\dot{A}}{A}\Bigg(\frac{e_1\lambda_2}{e_2}-\frac{\lambda_1}{2}\Bigg)+\frac{\dot{\lambda}_1}{2}-\frac{\dot{\lambda}_2e_1}{e_2}+\frac{\lambda_2}{2}\Bigg(\frac{\dot{e}_2}{e_2}-\frac{\dot{e}_1}{e_1}\Bigg)\right),
\label{eq:master_set1}
\end{eqnarray} 
where we have defined the following coefficients: $\lambda_1=b_1+Ab_3$, $\lambda_2=b_2+Ab_4$ and $\lambda_3=a_1+Aa_2+A^2a_3$.

For the form $J_\text{II}$ the solutions are the following
\begin{eqnarray} 
\mathcal{L}_c=X_{xx}=X_{xr}=X_{rr}=Z_{rr}=Y_{rx}=Y_{xr}=0, \nonumber\\
Y_{xx}=\frac{\dot{A}}{A}, \, Y_{rr}=\frac{\dot{B}}{B}. \nonumber \\ 
Z_{xx}=a_2(t) \dot A(t)+\frac{2 a_1(t) \dot A(t)}{A(t)}-A(t)^2 \nonumber\\-\dot a_3(t)-A(t) \dot{a}_2(t)-\dot a_1(t), \nonumber \\
Z_{xr}=\Imi\left(\frac{1}{2} \dot{\lambda}-\frac{\lambda}{2}\Bigg(\frac{\dot{A}}{A}+\frac{\dot{B}}{B}\Bigg)\right),
\label{eq:master_set2}
\end{eqnarray} 
where $\lambda=b_1+Ab_3+B(b_2+Ab_4)$.

Summarizing, letting the coefficients of forms $J_\text{I}$
and $J_\text{II}$ depend on time to
define one-parametric families, we have derived the condition for such families
to be smooth, see~\eref{eq:master_exist}.
Additionally we have found the explicit expressions for the ansatz coefficients
for both functional forms, see sets of equations (\ref{eq:master_set1}) and
(\ref{eq:master_set2}). In particular for $J_\text{II}$ the generator is
simpler than the one for $J_\text{I}$; this is expected because channels with
form $J_\text{II}$ depend on less parameters than $J_\text{I}$. Let us notice
that although singular channels do not admit master equations for the reasons
described above, it must not be confused with the fact that they can be reached
by smooth quantum processes.

\section{Conclusions} 
\label{sec:conclusions}
In this work we have critically reviewed the deceptively natural idea that
Gaussian quantum channels always admit a Gaussian functional form. To this end,
we went beyond the pioneering characterization of Gaussian channels with
Gaussian form presented in Ref.~\cite{PazSupplementary} in two new directions.
First we have shown that, starting from their most general definition as
mapping Gaussian states into Gaussian states, a more general parametrization of
the coordinate representation of the one-mode case exists, that admits
non-Gaussian functional forms. Second, we were able to provide a black-box
characterization of such new forms by imposing complete positivity (not
considered in Ref.~\cite{PazSupplementary}) and trace preserving
conditions. While our parametrization connects with the analysis done by
Holevo~\cite{Holevo2008} in the particular cases where besides having a
non-Gaussian form the channel is also singular, it also allows the study of
Gaussian unitaries, thus providing similar classification schemes. We completed
the classification of the studied types of channels by deriving the form of the
Liouvillian super operator that generates their time evolution in the form of a
master equation. Surprisingly, Gaussian quantum channels without Gaussian form
can be experimentally addressed by means of the celebrated Caldeira-Legget
model for the quantum damped harmonic oscillator~\cite{Gert}, where the new types of channels
described here naturally appear in the sub-ohmic regime.

\section*{Acknowledgements} 
We acknowledge PAEP and RedTC for financial support. Support by projects CONACyT 285754,
UNAM-PAPIIT IG100518 is acknowledged. CP
acknowledges support by PASPA program from DGAPA-UNAM. 
CAM and JDU acknowledge financial support from the German Academic Exchange
Service (DAAD). We are thankful to the University of Vienna where part of this
project was done. 
\appendix
\section{Explicit formulas for class $\mathbf{A_2}$} 
The explicit formulas of the final states for channels of class $A_2$ with the form
presented in \eref{eq:deltaop2} with $e_2=0$ are
\begin{eqnarray} 
\left(\sigma_f\right)_{11}&=\frac{1}{2 e_1},\nonumber\\
\left(\sigma_f\right)_{22} &= \left(\frac{\alpha}{\beta}\right)^2\left(\frac{b_3^2}{2 e_1}+2 a_3\right)+
\frac{\alpha}{\beta } \left(2 a_2+\frac{b_1 b_3}{e_1}\right)+\nonumber
& \   \ 2 a_1+\frac{b_1^2}{2 e_1}+s_1, \nonumber\\
\left(\sigma_f\right)_{12} &= -\frac{\alpha  }{\beta}\frac{b_3}{2 e_1}-\frac{b_1}{2 e_1},\nonumber\\
\vec d_f \left( s_3\right) &= \left(0, -\frac{\alpha}{\beta }c_2-c_1+s_2 \right)^\text{T},
\end{eqnarray} 
where
\begin{eqnarray} 
\fl
s_1=\left(b_2^2+2 \frac{\alpha}{\beta} b_2 b_4+\left(\frac{\alpha}{\beta}\right)^2b_4^2\right)\left(\sigma_i\right)_{11}\nonumber
-2\left(\frac{\alpha}{\beta}b_2+\left(\frac{\alpha}{\beta}\right)^2 b_4\right)\left(\sigma_i\right)_{12}
+ \left(\frac{\alpha}{\beta}\right)^2\left(\sigma_i\right)_{22},
\nonumber\\
s_2=\left(\frac{\alpha}{\beta} b_4+ b_2\right)(d_i)_1 -\frac{\alpha
}{\beta }(d_i)_2.
\end{eqnarray} 
The explicit formulas of the final states for channels of class $A_2$ with the
form presented in \eref{eq:deltaop2} with $\alpha=0$ are
\begin{eqnarray} 
\left(\sigma_f\right)_{11}&=\frac{e_2^2 }{4 e_1^2}\left(\sigma_i\right)_{11}+\frac{1}{2 e_1},\nonumber\\
\left(\sigma_f\right)_{12}&=\left(\frac{b_2 e_2}{2 e_1}-\frac{b_1 e_2^2 }{4 e_1^2}\right)\left(\sigma_i\right)_{11}-\frac{b_1}{2 e_1},\nonumber\\
\left(\sigma_f\right)_{22}&=2 a_1+\left(b_2-\frac{b_1 e_2}{2 e_1}\right)^2 \left(\sigma_i\right)_{11}+\frac{b_1^2}{2 e_1},
\end{eqnarray} 
and
\begin{equation}
\vec d_f=\left(\frac{e_2}{2 e_1} \left(\vec d_i \right)_1, \left(b_2-\frac{b_1 e_2}{2 e_1}\right)\left(\vec d_i \right)_1-c_1\right)^\text{T}.
\end{equation}
The explicit formulas of the final states for channels of class $A_2$ with Gaussian form
are
\begin{eqnarray} 
\left(\sigma_f\right)_{11}\left( s_1\right)&=\frac{2 a_3}{b_3^2}+s_1,\nonumber\\
\left(\sigma_f\right)_{12}\left(
s_1\right)&=\frac{a_2}{b_3}-\frac{2 a_3 b_1}{b_3^2}-b_1 s_1,\nonumber\\
\left(\sigma_f\right)_{22}\left( s_1
\right)&=\frac{b_1 \left(b_3 \left(b_1 b_3 s_1-2 a_2\right)+2 a_3 b_1\right)}{b_3^2}+2 a_1,\nonumber\\
\vec
d_f \left( s_2\right)&=\left(s_2-\frac{c_2}{b_3},b_1
\left(\frac{c_2}{b_3}-s_2\right)-c_1\right)^\text{T},
\end{eqnarray} 
where
\begin{eqnarray} 
s_1&=\frac{b_4^2}{b_3^2} \left(\sigma_i\right)_{11}-\frac{2 b_4 }{b_3^2}\left(\sigma_i\right)_{12}+\frac{1}{b_3^2}\left(\sigma_i\right)_{22},\nonumber\\
s_2&=\frac{1}{b_3}\left(d_i\right)_2-\frac{b_4
}{b_3}\left(d_i\right)_1.
\end{eqnarray} 

\bibliographystyle{unsrt}
\bibliography{labibliografia}

\begin{thebibliography}{10}

\bibitem{doi:10.1142/p489}
N.~J. Cerf, G.~Leuchs, and E.~S. Polzik.
\newblock {\em Quantum Information with Continuous Variables of Atoms and
  Light}.
\newblock Imperial College Press, 2007.

\bibitem{RevModPhys.86.1391}
M.~Aspelmeyer, T.~J. Kippenberg, and F.~Marquardt.
\newblock Cavity optomechanics.
\newblock {\em Rev. Mod. Phys.}, 86:1391--1452, Dec 2014.

\bibitem{RevModPhys.82.1041}
K.~Hammerer, A.~S. S\o{}rensen, and E.~S. Polzik.
\newblock Quantum interface between light and atomic ensembles.
\newblock {\em Rev. Mod. Phys.}, 82:1041--1093, Apr 2010.

\bibitem{Grosshans2003}
F.~Grosshans, G.~Van~Assche, J.~Wenger, R.~Brouri, N.~J. Cerf, and P.~Grangier.
\newblock Quantum key distribution using gaussian-modulated coherent states.
\newblock {\em Nature}, 421:238, Jan 2003.

\bibitem{PhysRevLett.82.1784}
S.~Lloyd and S.~L. Braunstein.
\newblock Quantum computation over continuous variables.
\newblock {\em Phys. Rev. Lett.}, 82:1784--1787, Feb 1999.

\bibitem{RevModPhys.77.513}
S.~L. Braunstein and P.~van Loock.
\newblock Quantum information with continuous variables.
\newblock {\em Rev. Mod. Phys.}, 77:513--577, Jun 2005.

\bibitem{PhysRevA.98.022335}
L.~Lami, B.~Regula, X.~Wang, R.~Nichols, A.~Winter, and G.~Adesso.
\newblock Gaussian quantum resource theories.
\newblock {\em Phys. Rev. A}, 98:022335, Aug 2018.

\bibitem{PazPhysRevLett}
E.~A. Martinez and J.~P. Paz.
\newblock Dynamics and thermodynamics of linear quantum open systems.
\newblock {\em Phys. Rev. Lett.}, 110:130406, Mar 2013.

\bibitem{Moshinsky1971}
M.~Moshinsky and C.~Quesne.
\newblock {Linear canonical transformations and their unitary representations}.
\newblock {\em J. Math. Phys.}, 12(8):1772--1780, 1971.

\bibitem{Ivan2011}
J.~S. Ivan, K.~K. Sabapathy, and R.~Simon.
\newblock {Operator-sum representation for bosonic Gaussian channels}.
\newblock {\em Phys. Rev. A}, 84(4):042311, Oct 2011.

\bibitem{Holevo2007}
A.~S. Holevo.
\newblock {One-mode quantum Gaussian channels: Structure and quantum capacity}.
\newblock {\em Probl. Inf. Trans.}, 43(1):1--11, mar 2007.

\bibitem{Cerf}
J.~Eisert and M.~M. Wolf.
\newblock {Gaussian Quantum Channels}.
\newblock In {\em Quantum Information with Continuous Variables of Atoms and
  Light}, pages 23--42. Imperial College Press, feb 2007.

\bibitem{Reviewquantuminfo}
C.~Weedbrook, S.~Pirandola, R.~Garc{\'{i}}a-Patr{\'{o}}n, N.~J. Cerf, T.~C.
  Ralph, J.~H. Shapiro, and S.~Lloyd.
\newblock {Gaussian quantum information}.
\newblock {\em Rev. Mod. Phys.}, 84(2):621--669, may 2012.

\bibitem{Gert}
H.~Grabert, P.~Schramm, and G.-L. Ingold.
\newblock {Quantum Brownian motion: The functional integral approach}.
\newblock {\em Phys. Rep.}, 168(3):115--207, oct 1988.

\bibitem{PazSupplementary}
Esteban~A. Martinez and Juan~Pablo Paz.
\newblock Dynamics and thermodynamics of linear quantum open systems.
\newblock {\em Phys. Rev. Lett.}, 110:130406, Mar 2013.
\newblock supplemental material.

\bibitem{cptp}
G.~Lindblad.
\newblock {Cloning the quantum oscillator}.
\newblock {\em J. Phys. A}, 33(28):5059--5076, 2000.

\bibitem{Holevo2008}
A.~S. Holevo.
\newblock {Entanglement-breaking channels in infinite dimensions}.
\newblock {\em Probl. Inf. Trans.}, 44(3):171--184, 2008.

\bibitem{Karrlein1997}
R.~Karrlein and H.~Grabert.
\newblock {Exact time evolution and master equations for the damped harmonic
  oscillator}.
\newblock {\em Phys. Rev. E}, 55(1):153--164, 1997.

\end{thebibliography}
\end{document}